%% file: main.tex
\documentclass{article}

\usepackage[margin=1.5in]{geometry}
\usepackage{flushend}
\usepackage{balance}

\usepackage{authblk}
\usepackage{amsmath}
\usepackage{amsthm}
\usepackage{amssymb}
\usepackage{url}
\usepackage[final]{showkeys}
\usepackage{tikz}
\usepackage{lineno}
\usepackage{graphics}
\usepackage{graphicx}

\usepackage{pgfplots}
\pgfplotsset{compat=1.4}
\usetikzlibrary{shapes}
\usetikzlibrary{plotmarks}

\usepackage{flushend}
\usepackage{framed}
\usepackage{mdframed}
\usepackage{placeins}
\usepackage{subcaption}
\usepackage{afterpage}
\usepackage{hyperref}
\usepackage{environ}
\newcommand{\eps}{\varepsilon}
\newcommand{\Oish}{\widetilde{O}}

\DeclareMathOperator{\dist}{dist}

\newtheorem{oq}{Open Question}
\newtheorem{claim}{Claim}
\newtheorem{definition}{Definition}
\newtheorem{theorem}{Theorem}
\newtheorem{lemma}{Lemma}
\DeclareMathOperator{\op}{\otimes}

\begin{document}

\title{The 4/3 Additive Spanner Exponent is Tight\footnote{Authors' emails: \{abboud, gbodwin\}@cs.stanford.edu.  This work was supported by NSF Grants CCF-1417238, CCF-1528078 and CCF-1514339, and BSF Grant BSF:2012338. This work was done (in part) while the authors were visiting the Simons Institute for the Theory of Computing.}}

\author{Amir Abboud}
\author{Greg Bodwin}

\affil{Stanford University}

\date{}

%\toappear{}

%\thispagestyle{empty}

\maketitle

\begin{abstract}
\quad A spanner is a sparse subgraph that approximately preserves the pairwise distances of the original graph.
It is well known that there is a smooth tradeoff between the sparsity of a spanner and the quality of its approximation, so long as distance error is measured {\em multiplicatively}.
A central open question in the field is to prove or disprove whether such a tradeoff exists also in the regime of \emph{additive} error.  That is, is it true that for all $\eps>0$, there is a constant $k_\eps$ such that every graph has a spanner on $O(n^{1+\eps})$ edges that preserves its pairwise distances up to $+k_{\eps}$?
Previous lower bounds are consistent with a positive resolution to this question, while previous upper bounds exhibit the beginning of a tradeoff curve:
all graphs have $+2$ spanners on $O(n^{3/2})$ edges, $+4$ spanners on $\Oish(n^{7/5})$ edges, and $+6$ spanners on $O(n^{4/3})$ edges. 
However, progress has mysteriously halted at the $n^{4/3}$ bound, and
despite significant effort from the community, the question has remained open for all $0 < \eps < 1/3$.

Our main result is a surprising negative resolution of the open question, even in a highly generalized setting.  We show a new information theoretic \emph{incompressibility} bound: there is no function that compresses graphs into $O(n^{4/3 - \eps})$ bits so that distance information can be recovered within $+n^{o(1)}$ error.
As a special case of our theorem, we get a \emph{tight} lower bound on the sparsity of additive spanners: the $+6$ spanner on $O(n^{4/3})$ edges cannot be improved in the exponent, even if any {\em subpolynomial} amount of additive error is allowed.  
Our theorem implies new lower bounds for related objects as well; for example, the twenty-year-old $+4$ emulator on $O(n^{4/3})$ edges also cannot be improved in the exponent unless the error allowance is polynomial.

Central to our construction is a new type of graph product, which we call the \emph{Obstacle Product}.
Intuitively, it takes two graphs $G,H$ and produces a new graph $G \op H$ whose shortest paths structure looks locally like $H$ but globally like $G$.

\end{abstract}

\input{intro}

\input{spannerlbs}

\input{tradeoff}

\input{conclusion}

\section*{Acknowledgments}

We are grateful to Virginia Vassilevska Williams for her sustained advice throughout this project, and to Seth Pettie and Ryan Williams for advice on an earlier draft of this paper that has considerably improved its presentation.
We thank Noga Alon and Uri Zwick for helpful comments.
We also thank an anonymous reviewer for useful corrections and writing suggestions.

\bibliographystyle{plain}
\bibliography{SpannersBIB}

\appendix
\input{ksumfree}
%\input{distpres}

\end{document}

%% file: intro.tex
\section{Introduction}

A \emph{spanner} of a graph is a sparse subgraph that approximately preserves the distances of the original graph.
Spanners were introduced by Peleg and Sch\"{a}ffer \cite{PS89} after they naturally arose in different contexts in the late 1980s \cite{Awe85,PU89,ADD+93}.
Today, spanners are an indispensable tool in many well-studied fields, such as graph compression, synchronization in distributed networks, routing schemes, approximation algorithms for all pairs shortest paths, and more.

% [orginally, multiplicative.]
The early results on the topic were mostly about multiplicative $t$-spanners, which are subgraphs that preserve all distances up to a multiplicative constant $t$.
A landmark upper bound result 
%(see also \cite{PS89, ADD+93, RTZ05, BK06}) 
states that every graph has a multiplicative $(2t-1)$-spanner on $O(n^{1+1/t})$ edges for all positive integers $t$ \cite{ADD+93}.
A well-known lower bound argument shows that this tradeoff is tight under the popular Girth Conjecture of Erd\"{o}s \cite{girth}.
Graphs on $\Omega(n^{1+1/t})$ edges with shortest cycles length $2t+2$ (which exist under the conjecture) cannot be sparsified at all without increasing a distance of a removed edge from $1$ to $2t+1$.
Therefore, the picture for multiplicative spanners is essentially complete.

% [additive is better but seemed too good to hope for. only results for restricted graphs.]
A \emph{$+k$ additive spanner} of a graph $G$ is a subgraph $H$ that preserves distances up to an {\em additive} constant $k$, i.e. for any two nodes $u,v$ in $G$ we have $\dist_G (u,v) \leq \dist_H(u,v) +k$.
In most contexts, additive error is more practically appealing than multiplicative error: a detour that adds two minutes to a trip is preferable to one that triples its length.
At first, this seemed hard to achieve; nontrivial additive spanners were only known for special types of graphs \cite{LS91,LS93}.
However, the seminal paper of Aingworth, Indyk, Chekuri, and Motwani \cite{ACM96,AICM99} showed the following surprise: all graphs have $+2$ additive spanners on just $\Oish(n^{3/2})$ edges (see also \cite{DHZ96,EP04,RTZ05,TZ06,Knu14,BKMP10}).
In other words, the {\em multiplicative} $3$-spanner shares a tight sparsity bound with the {\em additive} $+2$ spanner.
%[this result generated a lot of excitement and optimism and the main focus of the research community has been on additive spanners since then.]
This result caused significant optimism in the community regarding the existence of good additive spanners, which quickly became the central objects of study: could it be that graphs can be sparsified to near-linear sizes while incurring only a constant additive distance error?

The 20 years that followed witnessed a high throughput of results on additive spanners (see the recent Encyclopedia of Algorithms entry by Chechik \cite{ChechikSurvey}), and while this puzzle has been labelled
\begin{itemize}
\item ``major'' \cite{TZ06,CE06,Chechik13,ChechikSurvey,myspan,BV16}
\item ``main" \cite{BKMP05,BKMP10}
\item ``fascinating'' \cite{Woo06}
\item ``the chief open question in the field'' \cite{Pettie09}
\end{itemize}
and has been raised repeatedly, it has remained inexplicably wide open.

\begin{oq} \label{oq:main}
\emph{Prove or disprove: for all $\eps> 0$, there is a constant $k_{\eps}$ such that all graphs have a $+k_{\eps}$ additive spanner on $O(n^{1+\eps})$ edges.}
\end{oq}

Open Question \ref{oq:main} has been generally regarded with optimism.  This was motivated by some interesting new positive results: Chechik \cite{Chechik13} showed that all graphs have $+4$ spanners on $\Oish(n^{7/5})$ edges, and Baswana, Kavitha, Mehlhorn, and Pettie \cite{BKMP05,BKMP10} showed that all graphs have $+6$ spanners on $O(n^{4/3})$ edges (see \cite{Woo10,Knu14} for followup work).  This was also motivated by a lack of negative results: the Girth Conjecture implies an $\Omega(n^{1+1/k})$ lower bound for $+(2k-2)$ spanners, and the only further progress came from Woodruff \cite{Woo06}, who proved this lower bound unconditionally.
These bounds are fully compatible with a positive resolution of Open Question \ref{oq:main}.
Perhaps the most compelling argument for optimism, though, is the intuitive one: how could it be that spanners enjoy improved sparsity from a $+2$, $+4$, or $+6$ error allowance, but then suddenly the trail goes cold and no further tradeoffs are possible?
Despite all the evidence, the size upper bounds have not improved beyond $n^{4/3}$ for the last ten years, with no hint as to why.
%Despite all this evidence, upper bounds have mysteriously halted at the $n^{4/3}$ threshold for the last ten years, with no hint as to why this barrier should be at all difficult to surpass.

Meanwhile, considerable research effort has been spent on several promising directions of attack, with the explicit stated goal of making progress on this problem. 
Elkin and Peleg \cite{EP04} showed that there are near-linear size spanners with constant additive error in addition to a $(1+\eps)$ multiplicative error (see also \cite{Elkin05,TZ06,BKMP10} for more work on these \emph{mixed spanners}). 
Interesting spanners with non-constant additive error have been discovered \cite{BCE03,TZ06,Pettie09,BKMP10,Chechik13,myspan,BV16}.
Another example is the study of \emph{pairwise distance preservers} and \emph{pairwise spanners}, where we only require that distances between a small set of node pairs be (approximately) preserved \cite{BCE03,CE06,Pettie09,CGK13,KV13,Parter14,Kavitha15,BV16,AB16,B17}.
A lot of progress has also been done on additive spanners for special graph classes, see \cite{CDY05,DFG11,D+12,C+12,DA14} and the references therein.
While these fields have steadily progressed, none have proven sufficient to explain or beat the $n^{4/3}$ bound, and so the puzzle of Open Question \ref{oq:main} has endured.

%In a survey article on additive spanners by Chechik, the only other open question besides {oq:main} is: ``proving or disproving the existence of a spanner of size $O(n^{4/3-\eps})$ for some constant $\eps$ with constant or even polylog additive stretch would be a major breakthrough."

\subsection{Our Results}

Our first result is a negative resolution of Open Question \ref{oq:main}, and a negative resolution of the even stronger question of the tightness of the $n^{4/3}$ threshold.  We prove:
\begin{theorem}
\label{thm:main}
For all $\eps>0$, there is a $\delta>0$ and an infinite family of $n$ node graphs $G = (V, E)$ such that for any subgraph $H = (V, E')$ with $|E'| = O(n^{4/3-\eps})$, there exist nodes $u, v \in V$ with
$$\dist_H(u, v) = \dist_G(u, v) +\Omega(n^{\delta}).$$
\end{theorem}

Thus, the edge sparsity of $n^{4/3}$ cannot be improved, even for additive error $n^{o(1)}$,
%the $n^{4/3}$ spanner bound cannot be broken even for any additive error that is {\em subpolynomial} in the graph size, 
and the $+6$ spanner of Baswana et al. on $O(n^{4/3})$ edges is essentially the sparsest possible constant error additive spanner.
%construction until one considers the relaxation of additive spanners with $+n^c$ error.

%\paragraph{All hope is lost.}
Spanners are merely one (very appealing) form of graph compression. 
There are many other well-studied methods of compression which do not restrict the compressed form to be a subgraph of the original graph $G$; examples include Distance Oracles, Sketches, Labeling Schemes, Metric Embeddings, Emulators, etc.
For example, an \emph{emulator} is a sparse weighted graph $H$ (not necessarily a subgraph of $G$) that approximately preserves the distances of $G$ \cite{DHZ96,TZ06,Woo06,myspan}.

This leads to the following information theoretic question, which is an extremely relaxed version of Open Question \ref{oq:main}: 
Does there exist an algorithm for compressing a metric defined by an unweighted undirected graph into $O(n^{4/3-\eps})$ bits, for some $\eps>0$, such that approximate distances can be recovered from the compressed bitstring within constant additive error?
We further generalize Theorem \ref{thm:main} to show that our graph family is strongly {\em incompressible}, and provide a negative answer to this question as well.

\begin{theorem}
\label{thm:incompress}
For all $\eps > 0$, there exists a $\delta > 0$ such that there is no mapping $\psi$ from graphs $G$ on $n$ nodes to bitstrings of length $O(n^{4/3 - \eps})$ such that the distances of $G$ can always be recovered from the bitstring $\psi(G)$ within $+n^{\delta}$ error.
\end{theorem}
%In fact, we generalize this theorem even further to show that it holds even if only the distances within a node subset of size $n^{2/3 - \eps}$ are considered.

An interesting consequence of Theorem \ref{thm:incompress} is that the 20-year-old $+4$ additive emulator on $O(n^{4/3})$ edges of Dor, Halperin, and Zwick \cite{DHZ96} cannot be improved in the exponent even with any $n^{o(1)}$ additive error.

Finally, our resolution of Open Question \ref{oq:main} shows that polynomial additive distance error must be suffered in order to obtain near-linear size compression of graphs.
Given this, it is natural to wonder \emph{how much} polynomial error is necessary to obtain compression in this regime.
We show:
\begin{theorem}
For all $\eps > 0$, there exists a $\delta > 0$ such that there is no mapping $\psi$ from graphs $G$ on $n$ nodes to bitstrings of length $O(n^{1 + \delta})$ such that the distances of $G$ can always be recovered from the bitstring $\psi(G)$ within $+n^{1/22 - \eps}$ error.
\end{theorem}
The current best method for compression into $n^{1 + o(1)}$ bits obtains $+\Oish(n^{1/3})$ error \cite{myspan}.
We leave it as an open question to close this gap.

%% file: spannerlbs.tex
\section{The Construction}

The goal of this section is to explicitly construct the graphs mentioned in our theorems.
First, we offer a technical overview that gives the intuition behind our construction, which will help de-mystify some steps in our proof.
We will also highlight the novelty of our approach over previous lower bound constructions.

\subsection{Technical Overview}

Suppose we seek a lower bound against $+k$ spanners (think of $k$ as a small polynomial, like $n^{0.01}$ for now).  Our general approach is to start with a graph $G$ with the following special property: there is a set of node pairs $P$, with each pair at distance $k$, such that there is a unique shortest path in $G$ between each pair and the edge set of $G$ is precisely the union of these paths.
Thus, any edge deletion from $G$ will stretch one of the pairwise distances in $P$ by at least $+1$.
We then perform some transformations to $G$ to amplify this error.
The main trick here is to pick a family of ``obstacle graphs'' $H$, and then perform a new kind of replacement product $G \op H$, which we call the Obstacle Product (OP).
The effect will be that any spanner that deletes too many edges from $G \op H$ must increase the distance between one of the (transformed) node pairs in $P$ by at least $+k$.

Think of $G$ as the \emph{outer graph} and of $H$ as the \emph{inner graphs}.
The Obstacle Product $G \op H$ consists of two steps:
\begin{enumerate}
\item ({\em Edge Extension}) Replace every edge in $G$ with a length $k$ path
\item ({\em Inner Graph Replacement}) Replace each {\em original} node (i.e. not nodes created by the edge extension step) $v \in V(G)$ with an appropriately-chosen subgraph $H_v$ from the family of inner graphs.
\end{enumerate}

Finally, for each edge $\{u, v\}$ incident on a node $v$, one must choose exactly one node in $H_v$ that will serve as the new endpoint of this edge after the inner graph replacement step.  We will temporarily skip over the technical detail of how to choose this node.  The key outcome of the OP is that the (transformed) shortest paths for the node pairs in $P$ must now wander through many ``obstacles'' $H_v$ before reaching their final destination.

Our next step is to use a counting argument to say that for a sparse enough subgraph $F \subseteq G \op H$, there exists a (tranformed) pair in $\{s, t\} \in P$ such that $F$ is missing at least one of the edges used by the shortest path $\rho_{G \op H}(s, t)$ in {\em every} obstacle graph $H_v$ that this path visits.  We can then prove that
$$\dist_F(s, t) \ge \dist_{G \op H}(s, t) + k$$
Our proof of this claim uses two cases.  Possibly (1) the shortest paths $\rho_F(s, t)$ and $\rho_{G \op H}(s, t)$ follow the same path in the outer graph; that is, they intersect the same set of inner graphs in the same order.  In this case, because $F$ is missing an edge used by $\rho_{G \op H}(s, t)$ in every inner graph, $\rho_F(s, t)$ must take a $+1$ length detour at every inner graph, for a total error of $+k$.  Alternatively, (2) $\rho_F(s, t)$ and $\rho_{G \op H}(s, t)$ follow different paths in the outer graph; that is, $\rho_F(s, t)$ intersects a new inner graph.  In this case, we argue that $\rho_F(s, t)$ must travel an extra $k$-length path in the outer graph to reach its final destination, and this is the source of its $+k$ error.  See Figures \ref{fig:overview} and \ref{fig:overviewa} for a depiction of these two cases.

\begin{figure*}[t]
\begin{center}
%\begin{framed}
\begin{tikzpicture}[scale=0.6]

%overbrace
\draw [thick, decorate, decoration={brace}] (-6, 1.5) -- (6, 1.5);
\node [above=0.2cm] at (0, 1.5) {$\mathbf{k}$};

%connection
\draw [smooth, black, thick] plot [smooth] coordinates {(-7.5, 0) (-6, 0.5) (-4, -0.5) (-1, 0)};
\draw [smooth, black, thick] plot [smooth] coordinates {(1, 0) (4, 0.5) (6, -0.5) (7.5, 0)};

%inners
\draw [ultra thick, fill=white] (0, 0) circle [radius=1cm];
\draw [ultra thick, fill=white] (-5, 0) circle [radius=1cm];
\draw [ultra thick, fill=white] (5, 0) circle [radius=1cm];
\node [below=0.7cm] at (-5, 0) {$\mathbf{H_a}$};
\node [below=0.7cm] at (0, 0) {$\mathbf{H_b}$};
\node [below=0.7cm] at (5, 0) {$\mathbf{H_c}$};

%terms
\draw [fill=black, right=0.1cm] (-6, 0.5) circle [radius=0.15cm];
\draw [fill=black, left=0.1cm] (-4, -0.5) circle [radius=0.15cm];
\draw [fill=black] (-1, 0) circle [radius=0.15cm];
\draw [fill=black] (1, 0) circle [radius=0.15cm];
\draw [fill=black, right=0.1cm] (4, 0.5) circle [radius=0.15cm];
\draw [fill=black, left=0.1cm] (6, -0.5) circle [radius=0.15cm];

%start/end
\draw [fill=black] (-7.5, 0) circle [radius=0.15cm];
\node [above=0.2cm] at (-7.5, 0) {$\mathbf{s}$};
\draw [fill=black] (7.5, 0) circle [radius=0.15cm];
\node [above=0.2cm] at (7.5, 0) {$\mathbf{t}$};

%inner paths
\draw [dotted, ultra thick] (-6, 0.5) -- (-4, -0.5);
\draw [dotted, ultra thick] (-1, 0) -- (1, 0);
\draw [dotted, ultra thick] (4, 0.5) -- (6, -0.5);

\node [ultra thick, rotate=-30, above] at (-5, 0) {$\mathbf{+1}$};
\node [ultra thick, rotate=0, above] at (0, 0) {$\mathbf{+1}$};
\node [ultra thick, rotate=-30, above] at (5, 0) {$\mathbf{+1}$};
\end{tikzpicture}
\caption{\textbf{ \label{fig:overview}Case 1: Perhaps the new shortest path $\rho_F(s, t)$ still passes through the same set of inner graphs as $\rho_{G \op H}(s, t)$.  In this case, we use a counting argument to show that (for some pair $(s, t)$) the subgraph $F$ is missing an edge used by $\rho_{G \op H}(s, t)$ in every inner graph that it touches.  Therefore, $\rho_F(s, t)$ must take a $+1$ edge detour in every inner graph it touches, for a total of $+k$ error.}}
\end{center}
\end{figure*}
\vspace{1mm}
\begin{figure*}
\begin{center}
\begin{tikzpicture}[scale=0.6]
%connection
\draw [smooth, black, ultra thick, dashed] plot [smooth] coordinates {(-7.5, 0) (-6, 0.5) (-4, -0.5) (-1, 0)};
\draw [smooth, black, ultra thick, dashed] plot [smooth] coordinates {(1, 0) (4, 0.5) (6, -0.5) (7.5, 0)};

%inners
\draw [ultra thick, fill=white] (0, 0) circle [radius=1cm];
\draw [ultra thick, fill=white] (-5, 0) circle [radius=1cm];
\draw [ultra thick, fill=white] (5, 0) circle [radius=1cm];
\node at (-5, 0) {$\mathbf{H_a}$};
\node at (0, 0) {$\mathbf{H_b}$};
\node at (5, 0) {$\mathbf{H_c}$};

%terms
\draw [fill=black, right=0.1cm] (-6, 0.5) circle [radius=0.15cm];
\draw [fill=black, left=0.1cm] (-4, -0.5) circle [radius=0.15cm];
\draw [fill=black] (-1, 0) circle [radius=0.15cm];
\draw [fill=black] (1, 0) circle [radius=0.15cm];
\draw [fill=black, right=0.1cm] (4, 0.5) circle [radius=0.15cm];
\draw [fill=black, left=0.1cm] (6, -0.5) circle [radius=0.15cm];

%start/end
\draw [fill=black] (-7.5, 0) circle [radius=0.15cm];
\node [above=0.2cm] at (-7.5, 0) {$\mathbf{s}$};
\draw [fill=black] (7.5, 0) circle [radius=0.15cm];
\node [above=0.2cm] at (7.5, 0) {$\mathbf{t}$};

%new path
\draw [->, ultra thick] (-7.5, 0) -- (-7.5, -3);
\draw [thick] (-7.5, -3) to[bend right] (7.5, 0);
\node [right, ultra thick] at (-7.5, -1.5) {$\mathbf{k}$};
\node [ultra thick] at (0, -3.3) {$(+0)$};
\end{tikzpicture}
\caption{\textbf{ \label{fig:overviewa} Case 2: Perhaps the new shortest path $\rho_F(s, t)$ passes through an inner graph that is {\em not} touched by the old shortest path $\rho_{G \op H}(s, t)$.  In this case, we argue that $\rho_F(s, t)$ must travel an extra $k$-length path to reach its destination, and this is the source of its $+k$ error.}}
\end{center}

\end{figure*}

We consider the abstraction of this approach to be the biggest leap in understanding provided by our work.
Several important prior constructions can be viewed within this framework: the previous lower bound constructions of Woodruff \cite{Woo06} and Parter \cite{Parter14} can both be viewed within the obstacle product framework, with the inner graph as a biclique and the outer graph as a certain type of ``layered clique.''  A previous lower bound construction by the authors \cite{AB16} used the same ``layered clique'' outer graph, but allowed for various inner graphs to be substituted in.  The construction in this paper is the first one that allows for the modular substitution of outer graphs.
The ability to choose both inner and outer graphs turns out to be quite powerful: in addition to its central role in this paper, this type of product has proved useful in followup work constructing lower bound graphs for various other problems related to sketching graph distances \cite{B17,ABP17}.

In this paper, flexibility in our choice of outer graph turns out to be powerful enough that we are able to prove our tight spanner lower bounds {\em while only ever using cliques for our inner graphs}.  With this in mind, it will simplify our paper greatly to proceed with the restriction that $H$ is a family of cliques.
Let us now return to the main construction.  With the simplification that all inner graphs are cliques, we will execute the second step of the obstacle product (henceforth, the {\em clique replacement} step) as follows: replace each node $v$ with a clique on $\deg(v)$ nodes, and connect each of the $\deg(v)$ edges entering $v$ to a different clique node.  Following the argument above, we now have that every subgraph $F$ must keep at least one ``clique edge'' per pair in $P$, or else it stretches the distance between one of these paths by $+k$.  {\em Assuming that no two pairs $p_1, p_2 \in P$ ever use the same clique edge}, this implies that any $+(k-1)$ additive spanner of our graph $G \op H$ must keep at least $|P|$ edges in total.  The clique edge used by a (transformed) pair $p \in P$ in a clique $H_v$ is determined by the 2-path that the (original) path $\rho_G(p)$ uses to enter and leave the node $v$.  Thus, our lower bound of $|P|$ on the spanner density is realized so long as the shortest paths for every pair of pairs in $P$ is originally 2-path disjoint; i.e. $\rho_G(p_1) \cap \rho_G(p_2)$ does not contain any 2-paths in $G$ (see Figure \ref{fig:overview2}).

\begin{figure*}[t]
\begin{center}
%\begin{framed}
\begin{tikzpicture}

%lines
%red
\draw [ultra thick] plot [smooth] coordinates {({2*sin(360*0/3)}, {2*cos(360*0/3)}) (0, 0) ({2*sin(360*1/3)}, {2*cos(360*1/3)})};

%blue
\draw [ultra thick, dashed] plot [smooth] coordinates {({2*sin(360*1/3)}, {2*cos(360*1/3)}) (0, 0) ({2*sin(360*2/3)}, {2*cos(360*2/3)})};

%green
\draw [ultra thick, dotted] plot [smooth] coordinates {({2*sin(360*2/3)}, {2*cos(360*2/3)}) (0, 0) ({2*sin(360*0/3)}, {2*cos(360*0/3)})};

%v
\draw [fill=black] (0, 0) circle [radius=0.15cm];
\node at (0.3, 0.2) {$\mathbf{v}$};

%12gon
\foreach \t in {0,...,2}{
	\draw [fill=black] ({2*sin(360*\t/3)}, {2*cos(360*\t/3)}) circle [radius=0.15cm];
}

\draw [->, ultra thick] (3, 0) -- (5, 0);

\begin{scope}[shift={(8, 0)}]

%lines
%red
\draw [ultra thick] plot [smooth] coordinates {({2*sin(360*0/12)}, {2*cos(360*0/12)}) ({(2 + 1/sqrt(2))/2*sin(360*(-1)/24)},{(2 + 1/sqrt(2))/2*cos(360*(-1)/24)}) ({1/sqrt(2)*sin(360*0/3)}, {1/sqrt(2)*cos(360*0/3)}) ({1/sqrt(2)*sin(360*1/3)}, {1/sqrt(2)*cos(360*1/3)}) ({(2 + 1/sqrt(2))/2*sin(360*(9)/24)},{(2 + 1/sqrt(2))/2*cos(360*(9)/24)}) ({2*sin(360*2/6)}, {2*cos(360*2/6)})};

%blue
\draw [ultra thick, dashed] plot [smooth] coordinates {({2*sin(360*4/6)}, {2*cos(360*4/6)}) ({(2 + 1/sqrt(2))/2*sin(360*(17)/24)},{(2 + 1/sqrt(2))/2*cos(360*(17)/24)}) ({1/sqrt(2)*sin(360*2/3)}, {1/sqrt(2)*cos(360*2/3)}) ({1/sqrt(2)*sin(360*1/3)}, {1/sqrt(2)*cos(360*1/3)}) ({(2 + 1/sqrt(2))/2*sin(360*(7)/24)},{(2 + 1/sqrt(2))/2*cos(360*(7)/24)}) ({2*sin(360*2/6)}, {2*cos(360*2/6)})};

%green
\draw [ultra thick, dotted] plot [smooth] coordinates {({2*sin(360*4/6)}, {2*cos(360*4/6)}) ({(2 + 1/sqrt(2))/2*sin(360*(15)/24)},{(2 + 1/sqrt(2))/2*cos(360*(15)/24)}) ({1/sqrt(2)*sin(360*2/3)}, {1/sqrt(2)*cos(360*2/3)}) ({1/sqrt(2)*sin(360*0/3)}, {1/sqrt(2)*cos(360*0/3)}) ({(2 + 1/sqrt(2))/2*sin(360*(1)/24)},{(2 + 1/sqrt(2))/2*cos(360*(1)/24)}) ({2*sin(360*0/6)}, {2*cos(360*0/6)})};

%v
\draw [fill=black] ({1/sqrt(2)*sin(360*0/3)}, {1/sqrt(2)*cos(360*0/3)}) circle [radius=0.15cm];
\draw [fill=black] ({1/sqrt(2)*sin(360*1/3)}, {1/sqrt(2)*cos(360*1/3)}) circle [radius=0.15cm];
\draw [fill=black] ({1/sqrt(2)*sin(360*2/3)}, {1/sqrt(2)*cos(360*2/3)}) circle [radius=0.15cm];
\draw [ultra thick](0, 0) circle [radius={1/sqrt(2)}];
\node at (0, 0) {$\mathbf{v}$};

%12gon
\foreach \t in {0,...,2}{
	\draw [fill=black] ({2*sin(360*\t/3)}, {2*cos(360*\t/3)}) circle [radius=0.15cm];
}
\end{scope}

\end{tikzpicture}
%\end{framed}
\end{center}
\caption{\textbf{\label{fig:overview2} If no two paths enter and leave $v$ in the same way, then after $v$ is replaced with a clique, all of these paths will use a different clique edge.}}
\end{figure*}

We now have a lower bound on the number of edges that the spanner must keep; our next step is to obtain a favorable upper bound on the number of nodes in the spanner.  The dominant cost here is from the edge extension step, and so the number of nodes in the spanner is roughly equal to $E(G) \cdot k$.  We are now able to state the properties that we want $G$ to have:
\begin{enumerate}
\item $G$ is the union of many shortest paths between a set of node pairs $P$, with the following properties:
\begin{itemize}
\item Each pair $p \in P$ has distance $k$
\item Each pair $p \in P$ has a unique shortest path between its endpoints
\item These shortest paths are $2$-path disjoint
\end{itemize}
Note that we want $P$ to be as large as possible, since we will ultimately obtain a lower bound of $|P|$ on the number of edges in any $+(k-1)$ spanner.

\item $G$ has as few edges as possible, because the number of nodes in the spanner (due to the Edge Extension step) will be $k \cdot |E(G)|$.
\end{enumerate}

This completes the technical overview.  In the next two subsections, we will describe how to obtain starting graphs $G$ with these properties in full technical detail.  After that, we will repeat the details of the obstacle product transformation more formally, and we will fully prove that this series of transformation has the claimed properties.

\subsection{Starting Point}

Our starting point is the following lemma:

\begin{lemma} \label{celb} For all $\eps > 0$, there is a $0 < \delta < \eps$, and an infinite family of $n$ node graphs $G = (V, E)$ and pair sets $P \subseteq V \times V$ with the following properties:
\begin{enumerate}
\item For each pair in $P$, there is a unique shortest path between its endpoints
\item These paths are edge disjoint
\item For all $\{s, t\} \in P$, we have $\dist_G(s, t) = \Delta$ for some value $\Delta = \Theta(n^{\delta})$
\item $|P| = \Theta(n^{2 - \eps})$
\end{enumerate}
\end{lemma}

Alon \cite{Alon01} constructed graphs that prove this lemma in his work on property testing. To make this paper self-contained, we have included a full proof in the appendix.  The lemma can also be shown by a slight modification of the work of Coppersmith and Elkin \cite{CE06} on distance preservers, who proved a version of this lemma with all properties except the third.  

An additional property of the graphs from this lemma, implied by the others, is that $|E| = \Omega(n^{2 - \eps + \delta})$.  This is {\em too dense} for our purposes; for technical reasons discussed in the overview, we need a sparser object to prove interesting results.  The subject of the next lemma is to modify these graphs to reduce the edge count.

\subsection{Path Packing}

Our next move is to prove the following modification of Lemma \ref{celb}:
\begin{lemma} \label{compressedlb}
For all $\eps > 0$, there is a $0 < \delta < \eps$, and an infinite family of graphs $G=(V,E)$ and pair sets $P \subseteq V \times V$ with the following properties:
\begin{enumerate}
\item For each pair in $P$, there is a unique shortest path between its endpoints
\item These paths are 2-path disjoint (i.e. no two paths share any pair of consecutive edges)
\item For all $\{s, t\} \in P$, we have $\dist_G(s, t) = \Delta$ for some value $\Delta = \Theta(n^{\delta})$
\item $|P| = \Theta(n^{2 - \eps})$
\item $|E| = O(n^{3/2})$
\end{enumerate}
\end{lemma}
There are two important differences between Lemmas \ref{celb} and \ref{compressedlb}: the paths have become $2$-path disjoint rather than edge disjoint, and the edge count has fallen from $\approx n^2$ to $\approx n^{3/2}$.  For this reason, we think of this step as {\em path packing}, as we are essentially packing the same number of shortest paths into many fewer edges while only slightly relaxing the paths' overlap properties.

We will prove this lemma by starting with a graph as described in Lemma \ref{celb}, and performing a transformation to this graph to give it the properties of Lemma \ref{compressedlb}.  With this in mind, one should {\em not} think of the parameters $n, \eps, \delta$ in Lemma \ref{compressedlb} as the same as those in Lemma \ref{celb}; despite having the same names, they will change across the transformation.

We will now describe our transformation at a high level.  Start with a graph $G=(V,E)$ and pair set $P \subseteq V \times V$ as described in Lemma \ref{celb}.  Our goal is to create a new graph $G' = (V', E')$ and pair set $P' \subseteq V' \times V'$ as described in Lemma \ref{compressedlb}.  We achieve this by defining $G'$ to be a product of $G_1$ and $G_2$ (which are two identical copies of $G$), and we define $P'$ to be a product of $P_1$ and $P_2$ (which are two identical copies of $P$, although $P_1$ contains node pairs from $G_1$ and $P_2$ contains node pairs from $G_2$).  Intuitively, each new pair $p \in P'$ corresponds to a pair of pairs $(p_1 \in P_1, p_2 \in P_2)$, and a walk between the endpoints of $p$ in $G'$ corresponds to a walk between the endpoints of $p_1$ in $G_1$ \emph{and} simultaneously a walk between the endpoints of $p_2$ in $G_2$.
To enable this, each node in $G'$ corresponds to a pair of nodes $(u \in V_{G_1}, v \in V_{G_2})$, and we carefully design the new edge set $E'$ such that an edge in $E'$ corresponds to a step in $G_1$ {\em or} a step in $G_2$, but not both (e.g. exactly one of the two indices will remain the same across any edge in $E'$).
The shortest paths for $P'$ will be forced to take alternating steps in $G_1$ and $G_2$.

From this, we can argue $2$-path disjointness of the shortest paths for $P'$ as follows.  Consider a length $2$ subpath of the shortest path for some pair $p' \in P'$.  These two edges correspond to a step in $G_1$ and a step in $G_2$; since shortest paths for $P_1$ in $G_1$ are completely edge disjoint (and the same for $P_2, G_2$), these edges {\em uniquely} determine pairs $p_1 \in P_1$ and $p_2 \in P_2$.  Therefore, this information is sufficient to uniquely determine the new pair $p' = (p_1, p_2)$ being considered.

We will now describe this construction in full detail.
\begin{itemize}
\item \emph{The nodes:} For each (possibly non-distinct) pair of nodes $u_1,u_2 \in V$ and index $i \in \{1, 2\}$ we add the triple $(u_1,u_2,i)$ as a new node to $G'$. That is, the node set of our new graph $G'$ will be defined as: 
$$V' = \left\{ (u_1,u_2,i) \ \mid \ u_1,u_2 \in V, \ i \in \{1, 2\} \right\}. $$
Semantically, the index $i$ dictates whether we are supposed to take a step in the graph $G_1$ (represented by the first coordinate) or the graph $G_2$ (represented by the second coordinate).

\item \emph{The edges:} To define the edges of $G'$ we first need to define a \emph{forwards direction} for each edge in $E$. 
These directions will be chosen so that the shortest path between any pair $\{s,t\} \in P$ can be thought of as a walk that only traverses edges in their forwards direction. 
We do this by fixing an arbitrary ordering $s \leadsto t$ for every pair $\{s,t\} \in P$, and then directing the edges of the unique shortest path between this pair from $s$ towards $t$.  In other words, for any edge $\{u,v\}$ on the path, we define the $u \to v$ direction to be {\em forward} iff $v$ is closer to $t$ than it is to $s$. 
Let $\overset{\to}{P}$ be the version of $P$ with an arbitrary ordering imposed on each pair, and
let $\overset{\to}{E}$ be the version of $E$ with an (ordered) edge $(u,v) \in E$ corresponding to each (unordered) edge $\{u,v\} \in E$ if $u \to v$ was defined as the forwards direction. 
A crucial observation here is that the forwards direction is well-defined and consistent for all the edges in $E$, due to the property that every edge in $E$ is on the shortest path for exactly one pair in $P$.

\qquad We are now ready to define the edge set $E'$ of $G'$. 
The edges will be defined differently depending on the index value $i$ of the node. 
For every node $x=(u_1, u_2, 1) \in V'$, we add an edge $\{x,y\}$ to $G'$ if and only if $y=(u'_1,u_2,2)$ and $(u_1, u'_1) \in \overset{\to}{E}$.
Additionally, for every node $x=(u_1,u_2,2) \in V'$, we add an edge $\{x,y\}$ to $G'$ if and only if $y=(u_1,u'_2,1)$ and $(u_2,u'_2) \in \overset{\to}{E}$.
Intuitively, this enforces that any interesting shortest path must alternate between taking a forwards step in the first coordinate, and then a forwards step in the second coordinate.
More formally, the edge set of our new graph is:
\begin{align*}
E' := &\left\{ \{( u_1 ,u_2 ,1), (u'_1,u_2,2) \} \mid (u_1,u'_1) \in \overset{\to}{E}, u_2 \in V \right\} \bigcup\\
&\left\{ \{( u_1 ,u_2, 2), (u_1,u'_2,1) \} \mid (u_2,u'_2) \in \overset{\to}{E}, u_1 \in V \right\}\\
\end{align*} 
%$$E' := \left\{ \{( u_1 ,u_2 ,1), (u'_1,u_2,2) \} \mid (u_1,u'_1) \in \overset{\to}{E}, u_2 \in V \right\} \bigcup \left\{ \{( u_1 ,u_2, 2), (u_1,u'_2,1) \} \mid (u_2,u'_2) \in \overset{\to}{E}, u_1 \in V \right\}$$

\item \emph{The pair set:} Our new pair set $P'$ will be composed of pairs of old pairs from $P$, and its definition will rely on the ordering of each pair that we fixed above. For each ordered, possibly non-distinct pair of pairs $(s_1, t_1), (s_2, t_2) \in \overset{\to}{P}$, we add the pair $\{(s_1, s_2, 1), (t_1, t_2, 1) \}$ to $P'$.
Intuitively, given pairs $p_1, p_2 \in \overset{\to}{P}$, the corresponding new pair dictates that $p_1$ must be traveled in $G_1$ and $p_2$ must be traveled in $G_2$.
The formal definition of $P'$ is:
%\begin{align*}
%&P' :=\\
%&\left\{ \{(s_1, s_2, 1), (t_1, t_2, 1) \} \mid (s_1, t_1), (s_2, t_2) \in \overset{\to}{P} \ \right\}
%\end{align*}
$$P' := \left\{ \{(s_1, s_2, 1), (t_1, t_2, 1) \} \mid (s_1, t_1), (s_2, t_2) \in \overset{\to}{P} \ \right\}$$
\end{itemize}

In order to prove Lemma \ref{compressedlb}, we will first define a path for each pair in $P'$ that alternates between taking a step in $G_1$ and a step in $G_2$ as described above.

\begin{definition} [The path $\rho$]
We define
%\begin{align*}
%\rho(\{(s_1, s_2, 1), (t_1, t_2, 1)\}) :=&\\
%(u_1=(s_1, s_2, 1), u_2, \dots, &u_{k-1}, u_k=(t_1, t_2, 1))
%\end{align*}
$$\rho(\{(s_1, s_2, 1), (t_1, t_2, 1)\}) := (u_1=(s_1, s_2, 1), u_2, \dots, u_{k-1}, u_k=(t_1, t_2, 1))$$
where
\begin{enumerate}
\item $u_i$ for odd $i$ is equal to $(v^1_{\lceil i/2 \rceil}, v^2_{\lceil i/2 \rceil}, 1)$, and
\item $u_i$ for even $i$ is equal to $(v^1_{i/2 + 1}, v^2_{i/2}, 2)$
\end{enumerate}
where $v^1_j$ is the $j^{th}$ node on the unique shortest path between $s_1$ and $t_1$ in $G$, and $v^2_j$ is the $j^{th}$ node on the unique shortest path between $s_2$ and $t_2$ in $G$.
\end{definition}
In other words, $\rho(p')$ is built iteratively by alternatingly changing the first coordinate along the unique shortest path from $s_1$ to $t_1$ in $G$, and then changing the second coordinate along the unique shortest path from $s_2$ to $t_2$ in $G$.  
To see that the $\rho(p')$ is well defined for all $p' \in P'$ (i.e. the path $\rho(p')$ exists in $G'$), note that by construction an edge from $(v^1_{\lceil i/2 \rceil}, v^2_{\lceil i/2 \rceil}, 1)$ to $(v^1_{\lceil i/2 \rceil + 1}, v^2_{\lceil i/2 \rceil}, 2)$ (for odd $i$) exists iff the edge $(v^1_{\lceil i/2 \rceil}, v^1_{\lceil i/2 \rceil + 1})$ is in $\overset{\to}{E}$, which follows from the fact that $(s_1, t_1) \in \overset{\to}{P}$.  The other case, where $i$ is even, follows from an entirely symmetric argument.

Also note that the length of $\rho(p')$ is exactly $2 \Delta$ for all $p' \in P'$. 
Our next step is to argue that $\rho(p')$ is the unique shortest path between its endpoints.

\begin{claim}
For all pairs $p' \in P'$, the path $\rho(p')$ is the unique shortest path between the endpoints of $p'$.
\end{claim}

\begin{proof}
Let $p' =: (s', t')$ where $s' = (s_1,s_2,1)$ and $t'=(t_1,t_2,1)$.
By construction, any $s' \leadsto t'$ path $\pi := (s', x_1, \dots, x_k, t')$ corresponds to a $s_1 \leadsto t_1$ path $\pi_1$ in $G_1$ and a $s_2 \leadsto t_2$ path $\pi_2$ in $G_2$, where $\pi_1$ and $\pi_2$ are given by the sequence of values taken by the first and second indices of the nodes $x_i$ (respectively).
Moreover, we have $|\pi| = |\pi_1| + |\pi_2|$.

Note that the path $\rho(p')$ is the unique $s' \leadsto t'$ path $\pi$ that corresponds to the unique shortest $s_1 \leadsto t_1$ path in $G_1$ and the unique shortest $s_2 \leadsto t_2$ path in $G_2$.
Thus, any other $s' \leadsto t'$ path $\tau(p')$ from $s'$ to $t'$ corresponds to a non-shortest $s_1 \leadsto t_1$ or $s_2 \leadsto t_2$ path (without loss of generality, assume the former).
We thus have $|\rho(p')| = 2\Delta$, while $\tau(p') \ge (\Delta + 1) + \Delta > 2\Delta$.
Hence, $\rho(p')$ is the unique shortest $s' \leadsto t'$ path.
\end{proof}

Finally, we observe that, from our construction, the paths $\rho(p'),\rho(p'')$ are $2$-path disjoint for any distinct $p',p'' \in P'$.

\begin{claim}
For any two distinct pairs $p' \ne p'' \in P'$, the paths $\rho(p'), \rho(p'')$ are 2-path disjoint.
\end{claim}
\begin{proof}
Consider two pairs $p' = \{(s_1',s_2',1), (t_1',t_2',1)\} \in P'$ and $p'' = \{(s_1'',s_2'',1), (t_1'',t_2'',1)\} \in P'$ for which there is a $2$-path $(a, b, c)$ that is a subpath of both $\rho(p')$ and $\rho(p'')$. We will show that $p'$ and $p''$ must be the same pair.

By definition of $\rho$, the subpath $(a, b, c)$ must have one of two forms:
either $a=(x_1,y_2,1), b=(x_1',y_2,2),c=(x_1',y_2',1)$, or $a=(x_1,y_2,2), b=(x_1,y_2',1), c=(x_1',y_2',2)$.
Assume we are in the first case, and the second case is symmetric.
Again by the definition of $\rho$, we must have that $(x_1,x_1') \in \overset{\to}{E}$ is on the unique $(s'_1,t'_1)$ shortest path \emph{and} on the unique $(s''_1,t''_1)$ shortest path in $G$.
Since the shortest paths between pairs in $P$ are edge disjoint in $G$, this implies that $(s_1',t_1') = (s''_1,t''_1)$.
Moreover, we have that $(y_2,y_2') \in \overset{\to}{E}$ is on the unique shortest paths for both $(s'_2,t'_2)$ and $(s''_2,t''_2)$, and so we also have that $(s_2',t_2') = (s''_2,t''_2)$.
Together, these imply that $p'=p''$. 
\end{proof}

We can now prove Lemma \ref{compressedlb}:
\begin{proof} [of Lemma \ref{compressedlb}]
It is immediate from the construction that $N := |V'| = \Theta(n^2)$, and that $|P'| = |P|^2 = \Theta(n^{4 - 2\eps}) = \Theta(N^{2 - \eps})$ (thus, the new value of $\eps$ is the {\em same} as the old value of $\eps$ used to create $G$ via Lemma \ref{celb}).  Additionally, a loose (but sufficient for our purposes) upper bound on $|E'|$ follows from the observation that each node can have at most $2n$ neighbors, and thus $|E'| = O(n^3) = O(N^{3/2})$.

We have demonstrated above that $\rho(p')$ is the unique shortest path for any pair $p' \in P'$, and that $\rho(p'_1)$ and $\rho(p'_2)$ are 2-path disjoint for any $p'_1 \ne p'_2 \in P'$.  We additionally have $|\rho(p')| = 2\Delta = \Theta(n^{\delta}) = \Theta(N^{\delta/2})$ (and so the new value of $\delta$ is half the old value of $\delta$ used to create $G$ via Lemma \ref{celb}).  Since the values of $\eps, \delta$ used to create $G$ satisfied $0 < \delta < \eps$, and our new value of $\eps$ has remained unchanged while our new value of $\delta$ has been halved, we then still have $0 < \delta < \eps$.
\end{proof}

\subsection{The Obstacle Product}

Our final move is to produce an {\em Obstacle Product} (OP) $G \op K$, where $G$ is a graph produced by Lemma \ref{compressedlb}, and $K$ is a clique.  For intuition on the obstacle product and why it is useful, we refer the reader to the overview at the beginning of this section.  Its ultimate purpose is to prove our main theorem, which we restate below for convenience:

\begin{theorem} \label{splb}
For all $\eps>0$, there is a $\delta>0$ and an infinite family of $n$ node graphs $G = (V, E)$ such that for any subgraph $H = (V, E')$ with $|E'| = O(n^{4/3-\eps})$, there exist nodes $u, v \in V$ with
$$\dist_H(u, v) = \dist_G(u, v) +\Omega(n^{\delta}).$$
\end{theorem}

Let us now fix an $\eps>0$.  Our next steps will be to choose an appropriate $G, K$, describe how to build the obstacle product $G \op K$, and then argue that this is sufficient to prove Theorem \ref{splb}.

\paragraph{Starting Graph} We start by applying Lemma \ref{compressedlb} with parameter $\alpha=\eps/2$, to obtain a graph $G=(V,E)$ on $|V|=n$ nodes, $|E|=O(n^{3/2})$ edges, and a pair set $P \subseteq V \times V$ of size $|P|=\Theta(n^{2-\alpha})$ such that:
the distance between any pair $\{s,t\} \in P$ is exactly $\Delta$, for some $\Delta = \Theta(n^{\beta})$ and $0<\beta<\alpha$ (which implies $0<\beta< \eps/2)$, and the shortest paths between pairs in $P$ are $2$-path disjoint.

We next take the OP of $G$ with a clique $K$.  This OP consists of the following two transformations:

\paragraph{OP -- Edge Extension}
First, we replace every edge in $G$ with a path of length $\ell=3\Delta$.\footnote{This is overkill: the proof still works if we simply set $\ell = \Delta$.  However, a few additional minor technical details must be observed to push this argument through, so we use $\ell = 3\Delta$ here to maintain a bit of simplicity.}

More formally, if the edge $e = \{u,v\} \in E$, then we add the nodes $(e,1),\ldots,(e,\ell)$ to $V'$ and add the edges $\{u, (e,1) \}, \{v, (e,\ell)\},$ and $\{ (e,i), (e,i+1)\}$ for all $i \in [\ell-1]$ to $E'$. The choice of which endpoint $(e,1)$ or $(e,\ell)$ we connect to $u$ or $v$ is arbitrary. 

Note that the total number of nodes in $G'$ after this transformation is $O(n^{3/2+\beta})$.

\paragraph{OP -- Clique Replacement}
Next, for each {\em original} node $v \in V$ (i.e. not nodes introduced by the edge extension step), we replace $v$ with a clique on $\deg_{G}(v)$ nodes, with each incoming edge connected to a unique node in the new clique.

More formally, let $I(v) \subseteq E$ be the set of edges incident to $v \in V$, and introduce a new node $(v,e)$ for each edge $e \in I(v)$.
These new nodes will be connected in a clique: for each $e_i,e_j \in I(v)$ we add the edge $\{ (v,e_i), (v,e_j)\}$ to $E'$.
The clique that replaces a node $v$ will be denoted $K_v$.
We will call an edge contained in one of these cliques a {\em clique edge}.
After this replacement, each node $(v,e)$ will still be attached to an endpoint of the path corresponding to the edge $e$, i.e. we will have an edge $\{(v,e), (e,\alpha)\}$ where $\alpha \in \{1,\ell\}$ depends on the (arbitrary) ordering we chose in the edge-extension step.

The number of nodes added in this clique replacement step is exactly $\sum_{v \in V} \deg_G(v) = 2|E|$, which is $O(n^{3/2})$.
An important feature of this step is that we have introduced many new edges to $G'$; we will implicitly discuss this in the proof of correctness.\\

Let $k=\Delta-1 = \Theta(n^{\beta})$. 
To complete the proof, we will argue that any subgraph of $G'$ with fewer than $O(n^{2-\alpha})$ edges must distort the distances by more than $+k$.

\paragraph{Proof of Correctness}

First, we will build a pair set $P' \subseteq V' \times V'$; to prove Theorem \ref{splb} it will be sufficient to only consider node pairs in $P'$.  
For each pair $\{s,t\} \in P$, let the nodes on the unique $(s,t)$-shortest path in $G$ be $s= u_0 \to u_1 \to \cdots \to u_{\Delta-1} \to u_\Delta=t$, and let the $\Delta$ edges on this path be denoted $e_i = \{u_{i-1}, u_i \}$ for all $i \in [\Delta]$ (these definitions use the feature of Lemma \ref{compressedlb} that all pairs in $P$ are at distance exactly $\Delta$ in $G$).  Consider the nodes $s'=(s,e_1), t' = (t,e_{\Delta}) \in V'$, and add $\{s', t'\}$ as a pair to $P'$.

We will next reason about the structure of the shortest path between $s'$ and $t'$ in $G'$.
The following $(s',t')$ path will exist by construction:
Starting from $s'=(s_1,e_1)$, walk the path in $G'$ that replaced the edge $e_1$ in $G$, reaching the node $(u_1,e_1)$.
Then, walk the clique edge $\{ (u_1,e_1),(u_1,e_2)\}$, which we will denote by $e^{s,t}_1$ since it is the first clique edge on the $(s',t')$-path.
Then, similarly, walk the path that replaced to the edge $e_2$, and then walk the clique edge $\{ (u_2, e_2), (u_2, e_3)\}$, and so on until we reach $t'$.
This walk will traverse each of the $\ell$-length paths that replaced our edges $e_i$ (as well as some clique edges). 
The $i^{th}$ clique edge crossed will be $e^{s,t}_i = \{ (u_i, e_i) , (u_i, e_{i+1}) \}$, for $i \in [\Delta-1]$.
Thus, this walk will reach the node $t' = (t,e_{\Delta})$ after it has taken exactly $\Delta$ paths of length $ \ell$, plus $(\Delta-1)$ clique edges, i.e. $D := \Delta \cdot \ell + (\Delta -1)$ edges in total.
Let $\rho(s',t')$ be the path in $G'$ defined by this walk, and let $C^{s,t}$ be the set of clique edges on this path.

The following claim is an integral part of our proof:

\begin{claim}
\label{Disjoint}
For any two distinct pairs $\{s_1,t_1\} \ne \{s_2,t_2 \} \in P$, the corresponding clique edge sets $C^{s_1,t_1}, C^{s_2,t_2}$ are disjoint.
\end{claim}

\begin{proof}
Assume towards a contradiction that $C^{s_1,t_1}, C^{s_2,t_2}$ share a clique edge $\{ (u,e_i) , (u,e_j)\}$ for some $u\in V$ and $e_i,e_j \in I_G(v)$.
This implies that the $2$-path in $G$ composed of $e_i$ and $e_j$ is a subpath of the unique shortest $(s_1,t_1)$-path in $G$ \emph{and} of the unique shortest path $(s_2,t_2)$-path.
By Lemma \ref{compressedlb} this can only happen if $\{s_1,t_1\} = \{s_2,t_2\}$, which is a contradiction.
\end{proof}

Next, we examine the structure of alternate short paths between $s$ and $t$.
\begin{claim}
\label{Faithful}
Let $\{s, t\}$ be a pair in $P$, let $\{s',t'\}$ be the corresponding pair in $P'$, and let $\rho'(s', t')$ be a (not necessarily shortest) $(s', t')$ path of length less than $D+ \Delta$.
Let $\mathcal{K} = (K_s, K_{v_1}, \ldots, K_{v_{\Delta-1}}, K_t)$ be the sequence of cliques created in the Clique Replacement step that intersect $\rho(s', t')$, and let $\mathcal{K'} = (K_s, K_{v'_1}, \ldots, K_{v'_{x-1}}, K_t)$ be the sequence of cliques that intersect $\rho'(s', t')$.  Then $\mathcal{K} = \mathcal{K'}$.
\end{claim}
\begin{proof}
Note that, by construction, the node sequence $(s, v_1, \dots, v_{\Delta-1}, t)$ is the unique shortest $(s, t)$ path in $G$.  Similarly, note that the node sequence $(s, v'_1, \dots, v'_{x-1}, t)$ is an $(s, t)$ path in $G$.  If these node sequences are identical, then the claim holds.  Assume towards a contradiction that the node sequences differ.  Then the latter sequence must be at least one element longer than the former sequence (i.e. $x > \Delta$).  Both $\rho(s', t')$ and $\rho'(s', t')$ must travel a path of length $\ell$ between any two cliques, and so the number of path edges traveled by $\rho'(s', t')$ is at least $\ell = 3\Delta$ more than the number of path edges traveled by $\rho(s', t')$.  Meanwhile, $\rho(s', t')$ walks exactly $\Delta$ clique edges, while $\rho'(s', t')$ walks exactly $x$ clique edges.  Since $x \ge \Delta$, this implies that the total length of $\rho'(s', t')$ is at least $+3\Delta$ longer than $\rho(s', t')$.  This is a contradiction, and so it must be the case that the original node sequences are equal, and so $\mathcal{K} = \mathcal{K}'$.
\end{proof}

This lets us prove the next claim, which is the key to lower bounding the number of edges in any additive $+k$ spanner (recall that we chose $k=\Delta-1$).

\begin{claim}
\label{CliqueEdges}
For any pair $\{s,t\} \in P$, any path of length less than $D+ k$ from $s'$ to $t'$ must use at least one of the clique edges $C^{s,t}$.
\end{claim}

\begin{proof}
Let $\rho'(s',t')$ be a path from $s'$ to $t'$ of length less than $D + k$.
By Claim~\ref{Faithful}, we know that $\rho'(s',t')$ must pass through exactly the same cliques in the same order as $\rho(s',t')$, which implies that it must walk exactly the same set of paths introduced in the Edge Extension step of the OP.
This implies that for all $i \in [\Delta]$, $\rho'(s',t')$ contains a path from $(u_i,e_i)$ to $(u_i,e_{i+1})$ as a subpath. Let the length of this subpath be $d_i$ for some $d_i \geq 1$, and note that $d_i=1$ if and only if $\rho'(s',t')$ uses the clique edge $e^{s,t}_i \in C^{s,t}$.
A direct calculation shows that the length of $\rho'(s',t')$ is equal to
$$ \ell (\Delta-1) + \sum_{i=1}^{\Delta-1} d_i = D + \sum_{i=1}^{\Delta-1} (d_i-1)$$
If $d_i > 1$ for all $i$, then this is equal to at least $D + (\Delta - 1) = D + k$, which contradicts the assumption that $\rho'(s', t')$ has length less than $D + k$.  Thus, $d_i = 1$ for some $i$, and so $\rho'(s', t')$ uses the clique edge $\{(u_i, e_i), (u_i, e_{i+1})\}$.
\end{proof}

We can now prove our main theorem from a direct counting argument.
%\begin{proof} [Proof of Theorem \ref{splb}]
Recall that our graph $G'$ has $N=O(n^{3/2+\beta})$ nodes, and that our pair set $P'$ is of size $|P'| = \Theta(n^{2-\alpha})$.
Assume towards a contradiction that there is a subgraph $H=(V',E_H)$ of $G'$ on $|E_H|=O(N^{4/3-\eps}) = O(n^{(3/2+\beta)(4/3-\eps)})=O(n^{2+\frac{4}{3}\beta -\frac{3}{2} \eps -\eps \beta })$ edges in which for all pairs $\{s',t'\} \in P'$ the distance in $H$ is no more than the distance in $G'$ plus $k$, i.e. at most $D+k$.
By Claim \ref{CliqueEdges}, for each such pair $\{s',t'\} \in P'$ at least one of the clique edges in $C^{s,t}$ must exist in $H$, and from Claim \ref{Disjoint} we know that all these edges are distinct. 
Therefore, $E_H$ must contain at least one distinct edge per pair in $P'$, for a total of $|E_H|=\Omega(n^{2-\alpha})$ edges.
These two bounds on $|E_H|$ imply that:
\begin{align*}
2-\alpha &\leq 2+\frac{4}{3}\beta -\frac{3}{2} \eps -\eps \beta\\
\left(\frac{3}{2} +\beta \right) \eps &\leq \alpha +  \frac{4}{3}\beta
\end{align*}
We additionally know that $0 < \beta < \alpha \leq \eps/2$.  We complete the proof by arguing that there is no possible setting of $\alpha, \beta$ that makes these inequalities hold simultaneously, and so we have a contradiction and the subgraph $H$ cannot exist.  To see this, we will substitute $\beta = 0$ into the left hand side of the inequality (giving the left hand side of the inequality its lowest conceivable value), and simultaneously substitute $\alpha = \beta = \frac{\eps}{2}$ into the right hand side of the inequality (giving the right hand side of the inequality its highest conceivable value).  We obtain: 
\begin{align*}
\left(\frac{3}{2} + 0 \right) \eps &\le \frac{\eps}{2} + \frac{4}{3} \cdot \frac{\eps}{2}\\
\frac{3}{2} \cdot \eps &\le \frac{7}{6} \cdot \eps\\
%(\frac{3}{2})\eps &\le (\frac{7}{6}) \eps\\
%\frac{\eps}{3} &\le 0
\end{align*}
which contradicts the fact that $\eps > 0$.  This implies Theorem \ref{splb} with parameter $\delta$ set to $\beta$ (since $k=\Omega(n^{\beta})$).
 
%\end{proof}

\subsection{Strong Incompressibility}

Finally, we provide some additional observations about our construction to prove our strong incompressibility results.  We will prove:
\begin{theorem} \label{complb}
For all $\eps > 0$, there exists a $\delta > 0$ such that there is no mapping $\psi$ from graphs $G$ on $n$ nodes to bitstrings of length $O(n^{4/3 - \eps})$ such that the distances of $G$ can always be recovered from the bitstring $\psi(G)$ within $+n^{\delta}$ error.
\end{theorem}
Our argument proceeds along the lines of the well known proof that the Girth Conjecture implies that graphs cannot be compressed into $o(n^{1+1/t})$ bits without incurring a multiplicative error of $(2t-1)$ \cite{TZ05}, which is based on the works of Matou{\v{s}}ek \cite{Mat96} and Bourgain \cite{Bou85}.

For any subset of pairs $T \subseteq P'$, we define $G_T=(V,E_T)$ to be the subgraph of $G'$ obtained by removing all clique edges $C^{s',t'}$ corresponding to all pairs $\{s',t'\}$ in $T$.
Note that there are $2^{|P'|}$ such subgraphs $G_T$.

\begin{claim}
For all $T \subseteq P'$ and $\{s',t'\} \in P' \setminus T$ we have that $\dist_{G_T}(s',t') \leq D$.
\end{claim}

\begin{proof}
The only edges that are missing from $G_T$ are clique edges corresponding to pairs in $T$.
By Claim \ref{Disjoint}, none of these clique edges are in $C^{s',t'}$ for any pair $\{s',t'\} \notin T$.
Therefore, the length $D$ path $\rho(s',t')$ remains in $G_T$.
\end{proof}

\begin{claim}
For all $T \subseteq P'$ and $\{s',t'\} \in T$ we have that $\dist_{G_T}(s',t') > D+k$.
\end{claim}

\begin{proof}
Any path from $s'$ to $t'$ in $G_T$ cannot use any of the clique edges in $C^{s,t}$. By Claim~\ref{CliqueEdges}, the length of such path cannot be $\leq D+k$.
\end{proof}

The above two claims show that there exists a set $\mathcal{G}$ of $2^{|P'|} = 2^{\Theta(N^{4/3-\eps'})}$ graphs (which are $G_T$ for each possible pair subset $T \subseteq P'$) on the same node set $V'$ of size $N$, as well as a pair set $P' \subseteq V' \times V'$ of size $|P'|=\Theta(n^{4/3-\eps'})$ (for some $\eps'>0$), such that the following condition is true:

For any two graphs $G_1,G_2 \in \mathcal{G}$ there is a pair of nodes $\{s',t'\} \in P'$ for which 
$$
\left| \dist_{G_1} (s',t') - \dist_{G_2}(s',t') \right| \geq k
$$
where $k = \Omega(N^{\delta})$.
This means that for any compression function $\psi$ that takes a graph $G$ on $n$ nodes and produces a bitstring of length less than $|P|$, there exist two different graphs $G_1, G_2 \in \mathcal{G}$ that map to the same bitstring.  Thus, when the distance between $s'$ and $t'$ is recovered from this bitstring, it will be at least $\pm \Omega(n^{\delta})$ from the correct value for either $G_1$ or $G_2$.  Therefore, no such compression function can output bitstrings of length $O(n^{4/3-\eps'})$ such that distances can be recovered within $+O(n^{\delta})$ error.
%\end{document}

%% file: tradeoff.tex
%\documentclass{article}
%
%\usepackage{amsmath, amsthm}
%\newtheorem{theorem}{Theorem}
%\newtheorem{lemma}{Lemma}
%
%\newcommand{\eps}{\varepsilon}
%\DeclareMathOperator{\dist}{dist}
%
%\begin{document}

\section{Lower Bounds for Linear Size Graph Compression}

We next calculate our lower bound on the amount of (polynomial) additive error in distances that must be suffered in order to compress a graph into near-linear space.
We show:

\begin{theorem} \label{complblinear}
For all $\eps > 0$, there exists a $\delta > 0$ such that there is no mapping $\psi$ from graphs $G$ on $n$ nodes to bitstrings of length $O(n^{1 + \delta})$ such that the distances of $G$ can always be recovered from the bitstring $\psi(G)$ within $+n^{1/22 - \eps}$ error.
\end{theorem}

We begin with the following graphs:
\begin{theorem} [\cite{CE06}]\footnote{Coppersmith and Elkin prove a host of lower bounds, parameterized by an integer $d \ge 2$.  To help alleviate some tedious algebra, we have instantiated their construction here with $d=3$, which optimizes our lower bound ($d=4$ optimizes equally well).} \label{thm:dplb}
For any $n$ and $p = O(n^{3/2})$, there is an infinite family of $n$ node graphs $G = (V, E)$ on $\Theta\left(n^{3/5} p^{3/5}\right)$
edges, as well as pair sets $P \subseteq V \times V$ of size $|P| = p(n)$, with the following properties:
\begin{enumerate}
\item For each pair in $P$, there is a unique shortest path between its endpoints
\item These paths are edge disjoint
\item The edge set of $G$ is precisely the union of these paths
\item \label{it:eqlen} For all $\{s, t\} \in P$, we have $\dist_G(s, t) = \Delta$ for some value $\Delta = \Theta(n^{3/5}p^{-2/5})$
\end{enumerate}
\end{theorem}
This construction can be found in section 4 of \cite{CE06} (the statement is also in section 1.2).
Note that their construction does not have property \ref{it:eqlen} above.
However, it is easy to enforce this property using the following folklore technique.
Given an instance $G = (V, E), P$ satisfying the first three properties of Theorem \ref{thm:dplb} but not necessarily the fourth, let $\Delta := \frac{|E|}{2|P|}$ be half the average length of a path.
For each pair $(s, t) \in P$, split its unique shortest path $\pi$ into as many edge-disjoint subpaths as possible of length exactly $\Delta$ each.
Remove $(s, t)$ from $P$ and add the endpoints of each such subpath as a new pair to $P$.
If $\Delta$ does not divide $|\pi|$ and so there is a ``remainder'' subpath at the end of length less than $\Delta$, we discard this subpath and delete all of its edges from $G$.
At the end of this process, it is clear that all four properties in Theorem \ref{thm:dplb} hold for the new instance $G' = (V, E'), P'$.
Additionally, note that (1) we delete at most $|P| \cdot \Delta = \frac{|E|}{2}$ edges from $G$ to produce $G'$, so the density of $G'$ has only changed by a constant factor; and (2) we have $|P| = \frac{|E|}{\Delta}$ and $|P'| = \frac{|E'|}{\Delta}$, and since $|E'| = \Theta(|E|)$ the size of $P'$ has also only changed by a constant factor.

Unlike Lemma \ref{celb}, the construction in \cite{Alon01} or the one given in the appendix do not suffice; one must use the construction in \cite{CE06} to obtain this particular tradeoff between the number of pairs $p$ and the density of the graph.

We start with a graph $G = (V, E)$ drawn from Theorem \ref{thm:dplb}, choosing $p = n^{11/9 + \delta}$ and so $G$ has $\Theta\left( n^{4/3 + 3\delta/5} \right)$ edges.
Our construction proceeds exactly as before; we feed this graph through the path packing step, the edge extension step, and the clique replacement step.
Looking forward, we will reserve $n$ for the number of nodes in the original graph $G$, and $N$ to describe the number of nodes in $G$ after all transformations have been applied.

By the exact same argument for incompressibility given earlier in this paper, we have:
\begin{lemma} \label{lem:lincomp}
There is no mapping $\psi$ from graphs $G$ on $N$ nodes to bitstrings of length $o(p^2)$ such that the distances of $G$ can always be recovered from the bitstring $\psi(G)$ within $+o(\Delta)$ error.
\end{lemma}
We now simply calculate these parameters to obtain Theorem \ref{complblinear}.
\begin{itemize}
\item We have $p = n^{11/9 + \delta}$
\item We have $\Delta = \Theta\left(n^{3/5}p^{-2/5} \right) = \Theta\left(n^{3/5}(n^{11/9 + \delta})^{-2/5} \right) = \Theta\left(n^{1/9 - 2\delta/5}\right)$.
\item Originally, we have $n$ nodes and $\Theta(n^{4/3 + 3\delta/5})$ edges in $G$.
After the path packing step, we have $n^2$ nodes and $\Theta(n^{7/3 + 3\delta/5})$ edges in $G$.
The edge extension step then introduces
$$\Delta \cdot \Theta(n^{7/3 + 3\delta/5}) = \Theta\left(n^{1/9 - 2\delta/5} \cdot n^{7/3 + 3\delta/5} \right) = \Theta\left(n^{22/9 + \delta/5} \right)$$
new nodes to $G$, while the clique replacement step introduces $\Theta(n^{7/3 + 3\delta/5})$ new nodes to $G$ (i.e. one node for each edge in $G$ prior to the edge extension step).
Note that $n^{22/9 + \delta/5}$ dominates $n^{7/3 + 3\delta/5}$ for sufficiently small $\delta > 0$.
Thus, if we choose $\delta$ small enough, then we have $N = \Theta\left(n^{22/9 + \delta/5} \right)$.
\item Note: we then have
$$\Delta = \Theta\left(N^{(1/9 - 2\delta/5) / (22/9 + \delta/5)}\right) = \Theta\left(N^{1/22 - O(\delta)}\right)$$
and
$$p^2 = \Theta\left( N^{(11/9 + \delta)/(22/9 + \delta/5)} \right)^2 = \Theta\left( N^{(22/9 + 2\delta)/(22/9 + \delta/5)} \right) = \Theta\left( N^{1 + 81\delta/100 - O(\delta^2)} \right)$$
\end{itemize}
We now prove Theorem \ref{complblinear} by directly substituting these computed values into Lemma \ref{lem:lincomp}.
Given some $\eps > 0$ (in the context of Theorem \ref{complblinear}), we first choose $\delta > 0$ (controlling the size of the pair set in the graph $G$ from Theorem \ref{thm:dplb}) small enough that
$$\Delta > N^{1/22 - \eps/2}$$
This is possible because $\Delta = \Theta\left(N^{1/22 - O(\delta)}\right)$.
Hence, the fact that distances cannot be recovered within $+o(\Delta)$ error (from Lemma \ref{lem:lincomp}) implies that distances cannot be recovered within $+N^{1/22 - \eps}$ error, as desired.

Meanwhile, the compression threshold in Lemma \ref{lem:lincomp} is
$$o(p^2) = o\left( N^{1 + 81\delta/100 - O(\delta^2)} \right).$$
Thus, if we choose $\delta$ sufficiently small and choose some constant $\delta' = 4\delta/5 > 0$, then one cannot compress graphs into bitstrings of size $O(n^{1 + \delta'})$ and still recover distances within the error bound $+N^{1/22 - \eps}$ given above.
This completes the proof of Theorem \ref{complblinear}.

%% file: conclusion.tex
\section{Conclusion}

Our work shows that graphs cannot be compressed into $O(n^{4/3 - \eps})$ bits while maintaining distances within a constant (or even subpolynomial) additive error.  We thus implicitly show that the $4/3$ exponent is tight for additive spanners and emulators, and unbeatable for any other notion of graph compression where distance error is measured additively.
We consider it particularly interesting that the analogous statement is false for multiplicative error; that is, many graph compression schemes are capable of obtaining near-linear sized compression with only a constant multiplicative distance error.

Our main technical contribution is the Obstacle Product, which seems to be a powerful framework for generating families of graphs that are hard to efficiently sketch while approximately maintaining distances.
Below, we mention some open questions from the realm of spanners that are left open after this work, and for which our framework could be applicable.

\paragraph{Completing the picture of constant error additive spanners}
The only remaining gap (in the exponent) in our understanding of constant error additive spanners is the $+4$ error case. 
The lower bounds allow for the possibility of $+4$ spanners on $O(n^{4/3})$ edges, while upper bounds have only realized $\Oish(n^{7/5})$ \cite{Chechik13}.
Note that this gap does not exist for emulators.
Closing this gap is perhaps the most natural and intriguing open question that remains.

Our work also leaves $n^{o(1)}$ gaps between upper and lower bounds; that is, our lower bound holds against spanners on $O(n^{4/3 - \eps})$ edges, but it is still conceivable that an upper bound of $O(n^{4/3 - o(1)})$ edges could be achieved for a small enough $o(1)$ factor.
Hence, the following question remains open: do all graphs have $+k$ spanners on $O(\frac{n^{4/3}}{\log{n}})$ edges, or maybe even $O(\frac{n^{4/3}}{\log^k{n}})$?
%Our current construction allows this; with a finer-grained analysis of the rate at which $\delta$ and $\eps$ approach $0$, we can show that our graphs lack spanners of {\em any} constant additive error on $n^{4/3} / 2^{\omega(\sqrt{\log n})}$ edges.
It would be interesting to close this $n^{o(1)}$ gap.

\paragraph{Other Kinds of Spanners} 
We have shown that $+\Omega(n^{\delta})$ additive error is required for spanners on $O(n^{4/3 - \eps})$ edges.  It is then natural to ask precisely how much polynomial error is required for sparse spanners; for example, what is the smallest constant $\delta$ such that all graphs have $+O(n^{\delta})$ spanners of nearly linear size? 
The current best upper bound on $\delta$ is $\frac{3}{7}$, due to \cite{BV16}.
For general compression (specifically emulators), the upper bound on $\delta$ is $\frac{1}{3}$, due to \cite{myspan}.
Finding the exact tradeoff between the error parameter $n^\delta$ and the sparsity $n^{4/3-\eps}$ is still open.
Related notions are \emph{sublinear error} spanners, where if the original distance is $d$ it must be preserved up to $d+o(d)$, and $(\alpha,\beta)$ \emph{mixed} spanners, where even $(1+\alpha)\cdot d + \beta$ is allowed.
The framework described in this paper was adapted by the authors and Pettie \cite{ABP17} to obtain nearly-tight lower bounds in these settings as well, although some notable gaps still remain here to be closed.
Finally, it would be interesting to understand the exact tradeoff for \emph{pairwise} spanners \cite{CGK13,CE06,KV13,Kavitha15}.
Both of our main results hold even if only the distances within a pair set of size $\Theta(n^{4/3})$ need to be preserved.  When only smaller pair sets are considered, it is known that one can obtain $+6$ spanners on $O(n^{4/3 - \eps})$ edges \cite{Kavitha15}, although the exact nature of the tradeoff between pair set size, error, and sparsity remains open.

%\paragraph{Pairwise Distance Preserver Lower Bounds.} Perhaps the most promising way to solve some of the above problems is by improving the $\delta, \eps$ tradeoff in the original graphs from our first lemma.  These graphs are called {\em pairwise distance preserver lower bounds}, and there is currently a sizeable knowledge gap regarding their optimal sparsity \cite{CE06,BV16}. 
%In particular, it is known that for any $|P| = n^2 / 2^{\omega(\sqrt{\log n})}$ we can construct a graph $G$ and pair set $P$ such that the pairs in $P$ have unique edge-disjoint shortest paths of length $\omega(1)$ (see \cite{CE06}, or the construction in Appendix~\ref{ksum}).
% Through our obstacle product construction, this statement directly translates into the $n^{4/3} / 2^{O(\sqrt{\log n})}$ lower bound against constant error additive spanners mentioned above.
%This leads to the following interesting open question: can one improve the bound on $|P|$ in the above statement all the way up to $o(n^2)$?  If so, it would imply the very interesting statement that constant error spanners on $o(n^{4/3})$ edges are not possible in general.

%\end{document}

%% file: ksumfree.tex
\section{Graphs with many long disjoint shortest paths}

In this section we provide a proof of Lemma \ref{celb} which is the starting point for our lower bound constructions.
One way to prove the Lemma is by simple modifications to the constructions of Coppersmith and Elkin \cite{CE06} for \emph{pairwise distance preserver} lower bounds.
Their construction relied on lemmas from discrete geometry regarding the size of the convex hull of points in a radius $r$ ball in the $d$-dimensional integer lattice.
We will provide a simpler proof of this lemma (our proof is less efficient in the tradeoff between $\eps$ and $\delta$, but will suffice for our purposes).
Our proof relies on generalizations of the well-known arithmetic progression free sets that have found countless applications in combinatorics and theoretical computer science.

\paragraph{$k$-Average-Free Sets}
We say that a set of integers $S \subseteq [N]$ is {\em $k$-Average-Free} if, for any collection of $k$ integers  that are not all equal, $x_1, \dots, x_k \in S$, we have
$$\sum \limits_{i=1}^k x_i \ne k \cdot \bar{x}$$
for all $\bar{x} \in S$.  In other words, $S$ does not contain the arithmetic mean of any $k$ of its elements, except for the trivial case when all of these elements are equal.  An adaptation of a construction by Behrend \cite{Beh46} proves the existence of $k$-average-free sets with surprisingly large density. 
%We give a proof in the full version of the paper.

\begin{lemma}
\label{ksum}
For all $\eps>0$ there is a $\delta>0$ such that for infinitely many integers $N$ there exist $k$-average-free sets of size $\Theta(N^{1-\eps})$, with $k = \Omega(N^{\delta})$.
\end{lemma}

\begin{proof}
Fix $\eps>0$, and let $d = \lceil 3/\eps \rceil \in \mathbb{N}$ and $\delta = 1/(2d^2) >0$.
Let $p$ be a parameter that can take any value in $\mathbb{N}$; we will construct a $k$-average-free set $A \subseteq [N]$ where $k,N,$ and $|A|$ are functions of $p$. 
%Later, we will argue that these functions prove the lemma.

Let $X = [p]^d$ be the set of all $d$-dimensional vectors whose coordinates are integers in $[p]$.
For any $r \in [ d p^2]$, let
$$X_r := \{x \in X \ \mid \ \|x\|^2_2 = r \}$$
i.e. $X_r$ is the set of vectors for which the square of the Euclidean norm is $r$.
Observe that for all $x \in X$ we have
$$1 \le \|x\|^2_2 \le dp^2$$
and so the sets $X_1,\ldots,X_{dp^2}$ form a partition of the set $X$.
Let $V=X_r$ be the largest of these sets.
Since $|X|=p^d$, we have that
$$|V| \geq \frac{p^d}{dp^2} = \frac{p^{d-2}}{d}.$$ 

We now claim that the set $V$ is the analogue of a vector-valued ``$k$-average-free set" for {\em any} $k \in \mathbb{N}$, in the sense that the average of any set of vectors (not all equal) in $V$ is not an element of $V$.
This follows simply from the fact that all vectors in $V$ have the same Euclidean norm, and so they form a strictly convex set in $\mathbb{R}^d$.
By definition of convexity, for any $x_1, \dots, x_k, \bar{x} \in V$, we have
$$\left\| \sum \limits_{i=1}^k x_i \right\|_2 < \| k \bar{x} \|_2$$
whenever $x_1, \dots, x_k, \bar{x}$ are not all equal; thus, $\sum \limits_{i=1}^k x_i \ne k\bar{x}$, and so $V$ is a $k$-average-free set of vectors.

Our next step is to convert the vector set $V$ into an integer set $A$ in a way that preservers summation.
We remark that this step is not strictly necessary; i.e. it would be possible to prove Lemma \ref{celb} using only a vector-valued average free set.
However, it will simplify the proofs that follow to consider integers rather than vectors.
Our transformation will be defined by the function $f : [p]^d \to [N]$, where $N = ((k+1)p)^d$, in which for any vector $v =(v_1,\ldots,v_d) \in [p]^d$ we interpret $v_1, \ldots, v_d$ as the digits of a number in base $(k+1)p$.  More formally, we define
$$f(v) := v_1 \cdot q^0 + v_2 \cdot q^1 + \cdots + v_d \cdot q^{d-1}$$
where $q=(k+1)p$.
Note that this sum cannot exceed $q^d$, and so we have $f(v) \in [N]$.
An important feature of our transformation is that there will be no carry bits that transfer between the ``coordinates" when summing $k$ integers that correspond to vectors, and so $f$ distributes over $k$-length sums of vectors in $V$.  In other words, for any set of $k$ vectors $v^1, \dots, v^k \in V$, we have
$$\sum \limits_{i=1}^k f(v^i) = \left( \sum \limits_{i=1}^k v^i_1 \right)q^0 + \dots + \left( \sum \limits_{i=1}^k v^i_d \right)q^{d-1}.$$
Note that $\sum \limits_{i=1}^k v^i_j \le \sum \limits_{i=1}^k p \le kp$, and so we have
$$\left( \sum \limits_{i=1}^k v^i_1 \right)q^0 + \dots + \left( \sum \limits_{i=1}^k v^i_d \right)q^{d-1} = f \left( \sum \limits_{i=1}^k v^i \right)$$

And thus we have $\sum \limits_{i=1}^k f(v^i) = f \left( \sum \limits_{i=1}^k v^i \right)$, and so because $V$ is a $k$-average-free set of vectors, it follows that the corresponding integer set $A = f(V)$ is a $k$-average-free set of integers.

The final step in our proof is to verify that the set $A$ has size $\Theta(N^{1 - \eps})$ when $k = \Omega(N^{\delta})$, as claimed.  The parameters will work out correctly when we set $k = \lfloor p^{\frac{\delta d}{1-\delta d}} \rfloor-1.$
Note that we now have
$$N = ((k+1)p)^d = \Theta \left( \left(p^{\frac{1}{1-\delta d}} \right)^d \right) = \Theta(p^{d/(1-\delta d)})$$
and so $k = \Omega(N^{\delta})$.
Lastly, we must show that $|A| = \Omega(N^{1 - \eps})$.  Recall that $|A| \geq \frac{p^{d-2}}{d}$ and that $d= \lceil 3/\eps \rceil$ is a constant, which implies that $|A| = \Omega(p^{d-2})$.  We must then show that
$$p^{d-2} = \Omega(N^{1 - \eps}) = \Omega((p^{d/(1 - \delta d)})^{1 - \eps}) = \Omega(p^{(1 - \eps) \cdot d / (1 - \delta d)})$$
Or, equivalently, that
\begin{align*}
d-2 &\geq (1-\eps) \cdot \frac{d}{1-\delta d}\\
 (d-2)(1-\delta d) &\geq (1-\eps) \cdot d\\
\eps d - 2 + (2 \delta d - \delta d^2) &\geq 0
\end{align*}
Substituting in our previous choices of $d \geq 3/\eps$ and $\delta = \frac{1}{2d^2}$, we can verify that this inequality holds:
\begin{align*}
\eps d - 2 + (2 \delta d - \delta d^2) \geq 3 -2 + ( 0 - 1/2) > 0.
\end{align*}
This shows that $|A| = \Omega(N^{1 - \eps})$.  To prove the claim that $|A| = \Theta(N^{1 - \eps})$, simply note that we may discard elements from $A$ at will without destroying its average freeness property.
\end{proof}

\paragraph{From $k$-Average-Free sets to graphs}
We are now ready to prove Lemma~\ref{celb} which we repeat here for convenience.
The key idea is to let each node represent an integer in a careful way such that the $k$-average-free property translates into the uniqueness of certain shortest paths.

\begin{lemma} For all $\eps > 0$, there is a $0 < \delta < \eps$, and an infinite family of graphs $G = (V, E)$ and pair sets $P \subseteq V \times V$ with the following properties:
\begin{enumerate}
\item For each pair in $P$, there is a unique shortest path between its endpoints
\item These paths are edge disjoint
\item For all $\{s, t\} \in P$, we have $\dist_G(s, t) = \Delta$ for some value $\Delta = \Theta(n^{\delta})$
\item $|P| = \Theta(n^{2 - \eps})$
\end{enumerate}
\end{lemma}

\begin{proof}
Fix $\eps>0$ and apply Lemma \ref{ksum} with parameter $\alpha = \eps /2 >0$ to obtain an infinite sequence of integers $N$ and $k$-average-free sets $A \subseteq [N]$ of size $|A|=\Omega(N^{1-\alpha})$, in which $k=\Omega(N^{\beta})$ for some $\beta>0$. We will assume that $0<\beta$ is no more than $\eps/10$ (otherwise, we may decrease $\beta$ without destroying the above properties).
Consider the infinite sequence of integers defined by $n=N \cdot k^2 = \Theta(N^{1+2 \beta})$, and we will describe how to construct a graph $G=(V,E)$ on $n$ nodes that will be in our family, using the corresponding $k$-average-free set $A \subseteq [N]$.\\

\emph{The Nodes.} For each integer $x \in [(k+1) \cdot N]$ and ``layer index" $j \in \{0,\ldots,k\}$, we define a node $(x,j)$ and add it to $G$.  In other words, we define
$$V := \{ (x,j) \mid x \in [(k+1)N], \ j \in \{0,\ldots,k\} \}.$$

\emph{The Edges.} The edge set of $G$ will be defined using the set $A$.
Our graph will be layered such that all edges have the form $\{ (x,j), (y,j+1)\}$.
For all $j \in \{0,\ldots,k-1\}$ and for each integer $a \in A$, we add an edge $\{ (x,j), (y,j+1)\}$ to $E$ for all integers $x,y \in [kN]$ such that $x + a = y$.
That is:
%\begin{align*}
%E := \{\{(x, j), (y, j+1)\} \ \mid \ j \in \{0, \ldots, k-1\}&\\
%\exists a \in A \ x + a = y& \}
%\end{align*}
$$E := \{\{(x, j), (y, j+1)\} \ \mid \ j \in \{0, \ldots, k-1\} \quad \exists a \in A \ x + a = y \}.$$

\emph{The Pair Set.} Next, we define our pair set $P \subseteq V \times V$.
For every integer $x \in [N]$ and integer $a \in A$, we add the pair $\{s,t\}$ to $P$ where $s = (x,0)$ and $t = (x+ k \cdot a, k)$.
An immediate observation is that $|P|= N \cdot |A| = \Omega(N^{2-\alpha})$.
Formally:
$$P := \{ \{(x, 0), (x + k \cdot a, k)\} \ \mid \ x \in [N], a \in A \}.$$

We now turn to reasoning about the shortest paths between the pairs in $P$.
First, note that for any pair $s = (x,0)$ and $t = (x+ k \cdot a, k)$ in $P$, the edge set
$$E^{s, t} := \{\{ (x+ (j-1) \cdot a, j-1) , (x+ j \cdot a, j) \} \ \mid \ j \in [k] \}$$
forms a path of length $k$ from $s$ to $t$. We will refer to this path as $\rho(s,t)$.
Since our graph is layered, any $(s, t)$ path must contain at least $k$ edges, and so $\rho(s, t)$ is a shortest path between $s$ and $t$.  Additionally, any other path $\rho'(s, t)$ of equal length must use an edge set of the form
%\begin{align*}
%E'^{s, t} = \{\{(y_{j-1}, j-1)& , (y_{j-1}+ a_j, j)\} \ \mid \ \\
%&j \in [k], \{y_j\} \in [kN], \{a_j\} \in A \}.
%\end{align*}
$$E'^{s, t} = \{\{(y_{j-1}, j-1) , (y_{j-1}+ a_j, j)\} \ \mid \ j \in [k], \{y_j\} \in [kN], \{a_j\} \in A \}.$$
Since these edges form a path from $s$ to $t$, we deduce that  $y_0=x$ and $y_k = x +k \cdot a$. 
Combining this with the equations for these $k$ edges, we obtain
$$ k \cdot a = \sum_{j=0}^{k-1} a_j .$$
Since $A$ is a $k$-average free set, this implies that $a = a_0 = \cdots = a_{k-1}$, and so $\rho'(s, t) = \rho(s, t)$.  Therefore, $\rho(s, t)$ is the unique shortest path between $s$ and $t$, with length $k$.

Next, we show that these shortest paths are edge disjoint.
Any pair is determined by a starting point $(x,0)$ for some $x \in [N]$ and an element $a \in A$.
For any such pair, all the edges on the shortest path have the form $\{(c,j), (c+a,j+1)\}$.
Therefore, two pairs that disagree on the element $a$ cannot share any edge on the shortest path.
Moreover, if two pairs agree on the element $a \in A$ but disagree on the starting point (i.e. one starts at $(x,0)$ and the other starts at $(y, 0)$), then the corresponding shortest paths cannot even share a node:
all nodes on the first path have the form $(x + j \cdot a, j)$, while nodes on the second path have the form $(y + j \cdot a, j)$.
Therefore, the only way two paths can share a node is if $a + j \cdot a = b + j \cdot a$ and so $a = b$.

To conclude the proof we bound the parameters of our construction.
Since $\alpha=\eps/2$ and $\beta<\eps/10$, we have that
%\begin{align*}
%|P|=& \Omega\left(N^{2-\alpha}\right) = \Omega\left(n^{\frac{2-\alpha}{1+2\beta}}\right) \\
%\ge& \Omega\left(n^{\frac{2 - \eps/2}{1 + \eps/5}} \right) = \Omega\left(n^{2 - \frac{9\eps}{2\eps + 10}} \right) \ge \Omega\left(n^{2-\eps}\right)
%\end{align*}
$$|P|= \Omega\left(N^{2-\alpha}\right) = \Omega\left(n^{\frac{2-\alpha}{1+2\beta}}\right)\ge \Omega\left(n^{\frac{2 - \eps/2}{1 + \eps/5}} \right) = \Omega\left(n^{2 - \frac{9\eps}{2\eps + 10}} \right) \ge \Omega\left(n^{2-\eps}\right)$$
where last inequality follows from the fact that $0 < \eps \le 1$.  In the statement of the lemma we actually require that $|P|=\Theta(n^{2-\eps})$ as opposed to $\Omega(n^{2-\eps})$, but this follows because picking any appropriately-sized subset of $P$, which will still satisfies all the other properties.
Finally, note that the distances are $k= \Theta(N^{\beta})$ which is $\Theta(n^{\delta})$ for $\delta = \frac{\beta}{1+2\beta}$.  Finally, this setting of $\delta$ straightforwardly implies that $0 < \delta < \beta < \eps$, so the condition $0 < \delta < \eps$ holds.%
\end{proof}
%\end{document}